\documentclass[
reprint,
superscriptaddress,
nobibnotes,
amsmath
,amssymb,
aps,
pra,
 longbibliography
]{revtex4-2}
\pdfoutput=1
\usepackage[utf8]{inputenc}
\usepackage[english]{babel}
\usepackage[T1]{fontenc}
\usepackage{amsmath}
\usepackage{hyperref}
\usepackage{mathrsfs}
\usepackage{pgfplots, pgfplotstable}

\usepackage{tikz}
\usepackage{physics}

\usepackage{amsfonts}
\usepackage[utf8]{inputenc}
\usepackage{grffile}
\usepackage{graphicx}
\usepackage{flafter}
\usepackage{siunitx}
\usepackage[nameinlink]{cleveref}
\usepackage{dcolumn}% Align table columns on decimal point
\usepackage{bm}% bold math
\usepackage{tikz}
\usetikzlibrary{quantikz}
\usepackage[caption=false]{subfig}

\usepackage{xcolor}
\usepackage[normalem]{ulem}
\usepackage{amsfonts,amssymb,amsthm, bbm,braket,mathtools}

\definecolor{mygreen}{rgb}{0.21, 0.73, 0.13}

\newtheorem{theorem}{Theorem}

\newtheorem{lemma}[theorem]{Lemma}

\newcommand{\braakett}[2]{\langle\!\langle #1|#2\rangle\!\rangle}
\newcommand{\kett}[1]{|{#1}{\rangle\!\rangle}}
\newcommand{\braa}[1]{{\langle\!\langle}{#1}|}

\newcommand{\mc}[1]{\mathcal{#1}}

\newcommand{\ct}{{}^\dagger}

\newcommand{\tn}[1]{^{\otimes#1}}

\DeclarePairedDelimiter\parens{\lparen}{\rparen}

\newcommand{\ot}{\otimes}

\renewcommand{\epsilon}{\varepsilon}

\newcommand{\T}{\mathrm{T}}

\begin{document}
\title{Mutual information fluctuations and non-stabilizerness in random circuits}

\author{Arash Ahmadi}
\email{a.ahmadi-1@tudelft.nl}
\affiliation{QuTech and Kavli Institute of Nanoscience, Delft University of Technology, Delft, the Netherlands}
\author{Jonas Helsen}
\affiliation{QuSoft and CWI, Amsterdam, The Netherlands}
\author{Cagan Karaca}
\affiliation{QuTech and Kavli Institute of Nanoscience, Delft University of Technology, Delft, the Netherlands}
\author{Eliska Greplova}
\affiliation{QuTech and Kavli Institute of Nanoscience, Delft University of Technology, Delft, the Netherlands}

\begin{abstract}
The emergence of quantum technologies has brought much attention to the characterization of quantum resources as well as the classical simulatability of quantum processes. Quantum resources, as quantified by non-stabilizeress, have in one theoretical approach been linked to a family of entropic, monotonic functions. In this work, we demonstrate both analytically and numerically a simple relationship between non-stabilizerness and information scrambling using the fluctuations of an entropy-based quantifier. Specifically, we find that the non-stabilizerness generated by a random quantum circuit is proportional to fluctuations of mutual information. Furthermore, we explore the role of non-stabilizerness in measurement-induced entanglement phase transitions. We find that the fluctuations of mutual information decrease with increasing non-stabilizerness yielding potentially easier identification of the transition point. Our work establishes a key connection between quantum resource theory, information scrambling and measurement-induced entanglement phase transitions.
\end{abstract}

\maketitle

\section{Introduction} 
Since the formulation of quantum theory, entanglement has been known to be one of the unique signatures of quantum theory \cite{Einstein1935}. While entanglement is one of the key features of distinguishing the quantum world from the classical, it is known that entanglement alone is insufficient for having an advantage from the \emph{computational}  point of view. In particular, the stabilizer states are a set of states that can be highly entangled and at the same time can be simulated efficiently on classical computers \cite{gottesman1997stabilizer,Nielsen_Chuang_2010,Gottesman_fault_tolerant_1998,Aaronson_gottestman_2004,Gottesman_error_correction_1996,gottesman1998heisenberg}. 

The missing feature for quantum computation to have an advantage over classical computation is a concept known as non-stabilizerness or \emph{magic} \cite{Veitch_2014}. In quantum circuit language, in the fault-tolerant regime, the gates that are known to be "cheap" gates \cite{Veitch_2014} are Clifford gates, and magic is injected into the state by adding non-Clifford gates to the circuit. It is also known that the Clifford operations could have easier implementation both at the experimental level and for quantum error correction \cite{Gottesman_fault_tolerant_1998,gottesman1997stabilizer,shor1997faulttolerant,Beverland_2020}. This feature makes magic the resource for quantum computation.

Various measures for non-stabilizerness, or magic, have been proposed. Some of the notable ones are magical cross-entropy, mana \cite{Veitch_2014}, robustness of magic \cite{Howard_2017,hamaguchi2023handbook}, stabilizer entropies (SE) \cite{Leone_2022} and others \cite{Haug_scalable_2023,Liu_manybody_2022,Garcia_resource_2023,turkeshi2024magicspreadingrandomquantum,garcia2024hardnessmeasuringmagic}. Due to the favourable scaling of the stabilizer entropies these measures received much attention from the community to study the properties of non-stabilizerness of quantum systems \cite{leone2023phase,Rattacaso_2023,Salvatore_topology_2022,leone2023learning,leone_fidelity_2023,Oliviero_ising_2022,haug_frustration_2023}. Additionally, there has been a significant effort to scale up the computability of stabilizer entropies \cite{haug2023efficient,haug_mps_2023,tarabunga2024nonstabilizerness,lami_mps_2023,Lami_2024}. 

Diverse studies confirmed there exists a relationship between non-stabilizerness and information scrambling \cite{ahmadi2024quantifying,Garcia_resource_2023,Leone_2022}. Information scrambling describes the spread of local information in a generic quantum system \cite{Hosur_2016}. Information scrambling has been shown to be one of the most powerful measures for various quantum properties of quantum systems, from black holes \cite{Hayden_2007,Sekino_2008,Shenker_2014,Swingle_2018} to many-body quantum systems \cite{BASKO20061126} to quantum circuits \cite{Mi_OTOC_2021,ahmadi2024quantifying}. There are a number of known methods for measuring information scrambling. One of the well-known approaches is based on the measurement of correlator functions, namely Out-of-Time Ordered Correlators (OTOCs) \cite{Hayden_2007,Swingle_2018,Mi_OTOC_2021,khemani_operator_2018,ALEINER2016378,robert_diagnosing_2015}. Another approach for studying information scrambling are entropy-based measures \cite{Alba2019scrambling, touil2024information,monaco2023operational,nahum_entanglement_2017,Hamma_mutual_2016,Hosur_2016,Schnaack_2019,Pappalardi_scrambling_2018}. 

\begin{figure}
    \centering
    \includegraphics[width=\linewidth]{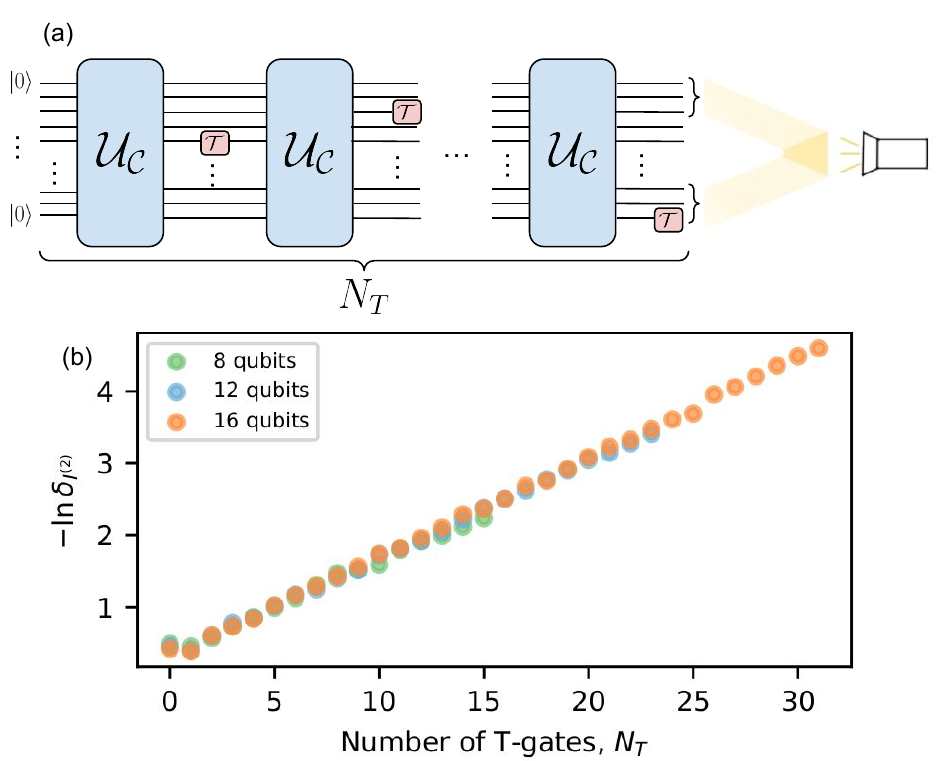}
    \caption{(a) The schematic structure of the t-doped Clifford circuits and the disjoint area of measuring scrambling. (b) The log of fluctuations of the mutual information, $-\ln
     \delta_{I^{(2)}}$ as a function of the number of T-gates, $N_T$, on the circuit for the number of qubits, $N=8,12,16$ and the number of samples is 500.}
    \label{fig:fig1}
\end{figure}

In this paper, we describe numerically and analytically a general relationship between non-stabilizerness and \emph{fluctuations} of entropy-based measures of information scrambling.
%namely mutual information, on t-doped Clifford circuits.
A specific instance of this relation has been observed earlier by connecting fluctuations of OTOCs to non-stabilizerness \cite{ahmadi2024quantifying,Leone_2022}. In the present work, we generalize this behaviour to fluctuations of generic mutual information measures of disjoint regions of the quantum circuit on t-doped circuits. 
The observed link between non-stabilizerness and mutual information creates a bridge to the theory of phase transitions in measurement-induced random quantum circuits, where mutual information plays a crucial role \cite{Li_2018,Li_2019,skinner_2019,Chan_2019,Jian_2020,Gullans_2020,Fisher_2023,Cao_2019,Bao_2020,Zabalo_2020,Tang_2020,Fuji_2020,Iaconis_2020,Lang_2020,Nahum_2021,Lavasani_Alavirad_Barkeshli_2021,Sang_2021,Lunt_2020,Szyniszewski_2020,Gullans_2021,Alberton_2021,Turkeshi_2020,Ippoliti_2021,Yaodong_2021,Fan_2021,Lavasani_2021,Regemortel_2021,Claeys_2022,Agrawal_2022,Block_2022,Noel_Niroula_Zhu_Risinger_Egan_Biswas_Cetina_Gorshkov_Gullans_Huse_etal._2022,Koh_Sun_Motta_Minnich_2023,Oshima_2023}. Here we show that non-stabilizerness is an ingredient in reducing the spread of entanglement phase transition in random quantum circuits and also reduces the fluctuations of the measured mutual information.

\section{Random circuits and mutual information}
\subsection{Definitions and notation}

The structure of the T-doped circuits that we use in this study consists of blocks of random Clifford operations followed by single T-gates on random qubits in the circuit. The blocks of random Clifford consist of single Clifford gates drawn randomly from the set $\{ I,X,Y,Z,H,S\}$ followed by three CNOT gates on two random qubits. We apply both single and double Clifford gates until the state gets fully scrambled (in our case after $2N$ operations, where $N$ is the number of qubits in the system). Each Clifford block (of depth $2N$) is followed by a T-gate performed on a randomly selected qubit. This Clifford block plus T-gate sequence is then repeated $N_T$ times. The structure of the circuit is shown in Fig. \ref{fig:fig1} (a).

A popular entropy-based measure for quantifying scrambling is mutual information \cite{Alba2019scrambling} of disjoint areas $A$ and $B$ which is defined as 

\begin{equation}
    I^{(2)} := S^{(2)}_A + S^{(2)}_B - S_{AB} ^{(2)},
    \label{eq:1}
\end{equation}
where $S^{(2)}_X$ is the Renyi-$2$ entropy defined as  $S^{(2)}_X \equiv - \log_2 \Tr \rho_X^2$ for the subsystem $X$ and $\rho_X$ is the reduced density matrix of subsystem X. We chose the subsystems $A$ and $B$ such that they have $N_A=N_B=N/4$ of the first and last qubits as shown in Fig. \ref{fig:fig1} (a).

\begin{figure}
    \centering
    \includegraphics[width=\linewidth]{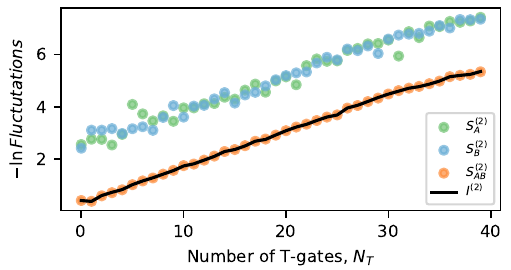}
    \caption{ The comparison of the $-\ln{fluctuations}$ of all terms in the mutual information definition, $S_A ^{(2)}$, $S_B ^{(2)}$, $S_{AB} ^{(2)}$, and the mutual information as a function of the number of T-gates, $N_T$ for a 16-qubit system.}
    \label{fig:fig2}
\end{figure}

\subsection{Relation of mutual information fluctuations and magic}
\label{sec:2B}

By measuring a sample of different instances of $I^{(2)}$ for a fixed number of T-gates, we observed a trend in the fluctuations, i.e. the standard deviation, of these instances
\begin{equation*}
    \delta_{I^{(2)}}= (\mathbb{E}[(I^{(2)})^2]-\mathbb{E}[I^{(2)}]^2)^{1/2},
\end{equation*}
 as a function of the number of T-gates, $N_T$, in the random circuit. We observe a linear relationship between the number of T-gates in the circuit and $-\ln \delta_{I^{(2)}}$ as shown in Fig. \ref{fig:fig1} (b). We have a circuit of $N=8,12,16$ qubits and the number of T-gates in each circuit is $N_T\in \{0,...,2N\}$ which is where magic increases linearly with the number of qubits \cite{ahmadi2024quantifying,Leone_2021}. We fit this linear behavior and obtain $-\ln \delta_{I^{(2}} \approx 0.13 N_T + 0.32$ dependency. Interestingly, this behaviour is similar to that previously observed for the fluctuations of out-of-time-order correlations~\cite{ahmadi2024quantifying}.

In order to interpret the results in Fig. \ref{fig:fig1} (b), we analyze Eq.~\ref{eq:1} term by term and asses the contribution of these terms to the observed fluctuation behaviour. Eq.~\ref{eq:1} have three terms of entanglement Renyi-$2$ entropy, where subsystems $A$ and $B$ have the size of $N/4$ and subsystem $AB$ has the size of $N/2$. By evaluating the fluctuations of $S_A ^{(2)}$, $S_B ^{(2)}$ and $S_{AB} ^{(2)}$ alongside the total fluctuations of the mutual information, we find that the fluctuations of $S_{AB} ^{(2)}$ are the leading term in behaviour of fluctuations of mutual information, $\delta_{I^{(2)}}$. This result is shown in Fig.~\ref{fig:fig2}. 

Specifically, in Fig. \ref{fig:fig2}, we show all contributions to the fluctuations of the mutual information for a 16-qubit system and observe the fluctuation of the mutual information overlaps with the fluctuations of the entanglement Renyi-$2$ entropy of the subsystem $AB$.

We also explored the effect of measuring the total spin of the subsystems $A$ and $B$ and their statistical relation to non-stabilizerness in Appendix \ref{sec:kurtosis}.

\begin{figure}
    \centering
    \includegraphics[width=\linewidth]{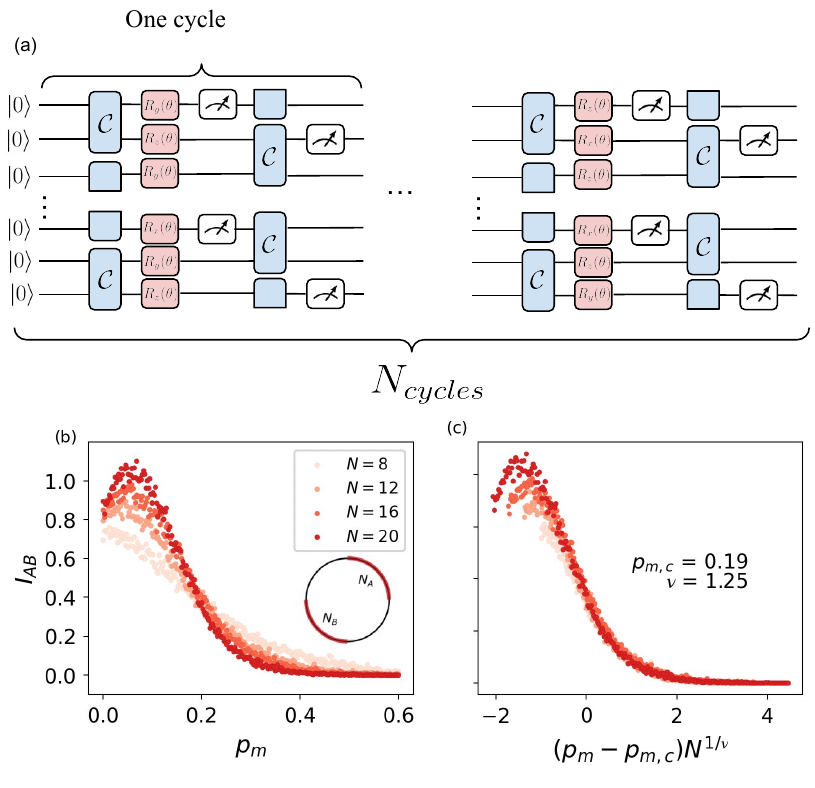}
    \caption{(a) The schematic structure of measurement-induced phase transition. The blue block shows the random two-qubit Clifford gate and the red single-qubit gates are random rotation gates followed by a projective measurement with probability $p_m$. The whole circuit consists of repeating each cycle $N_{cycle}$ times. (b) The mutual information of two disjoint partitions of the system, namely the first and third quarters ($N_B = N_B = N/4$), averaged over 800 instances and $N_{cycle}=125$ for different system sizes and the rotation angle, $\theta =0$. (c) The finite size scaling of the mutual information for the critical parameters $p_{m,c}=0.19$ and $\nu = 1.25$. }
    \label{fig:fig3}
\end{figure}

\subsection{Analytical relation of \texorpdfstring{$N_T$}{NT} and \texorpdfstring{$\ln(\delta_{I^{(2)}})$}{log(deltaI)}}
Due to the non-linear nature of the entropy, it is difficult to analytically recover the behaviour seen in Fig.~\ref{fig:fig1}. However if one instead averages \emph{inside} the logarithm exact calculation becomes tractable (this can be thought of as the first step towards a replica-trick calculation, which we will not attempt here) and we can explain the key features of the relation between magic and entropic fluctuations. Accordingly, we define the quantities
\begin{align*}
\tilde{S}_{AB}^{(2)} &:= -\log_2\big(\mathbb{E}\big[\Tr(\rho_{AB}^2)\big]\big),\\
\tilde{\delta} &:=\big(-\log_2\big(\mathbb{E}\big[\Tr(\rho_{AB}^2)^2\big]\big) -(\tilde{S}_{AB}^{(2)})^2\big),
\end{align*}
where the average is taken over the circuit set described in Fig.~\ref{fig:fig1}. A simple calculation shows that $\tilde{S}_{AB}^{(2)}$ is independent of $N_T$. On the other hand $\tilde{\delta}$ shows a clear linear dependence.  Suppressing various subleading contributions, we can prove the following relation:
\begin{equation}
-\ln(\tilde{\delta}) \approx N_T \ln(1/\lambda)
\label{eq:proof-eq}
\end{equation}
for $\lambda = 3/4$ and thus $\ln(1/\lambda) \approx  0.28$. This captures the linear behaviour seen in Fig.~\ref{fig:fig1} but does not give the correct rate. We expect that the correct rate can only be obtained through a full replica calculation, similar to the behaviour of the velocity of entanglement observed in \cite{zhou2019emergent}. The proof of Eq.~\eqref{eq:proof-eq} requires calculating the fourth moments of the Clifford group and is rather involved. We defer it to Appendix \ref{sec:analytics}{}. 
\begin{figure*}
    \centering
    \includegraphics[width=\textwidth]{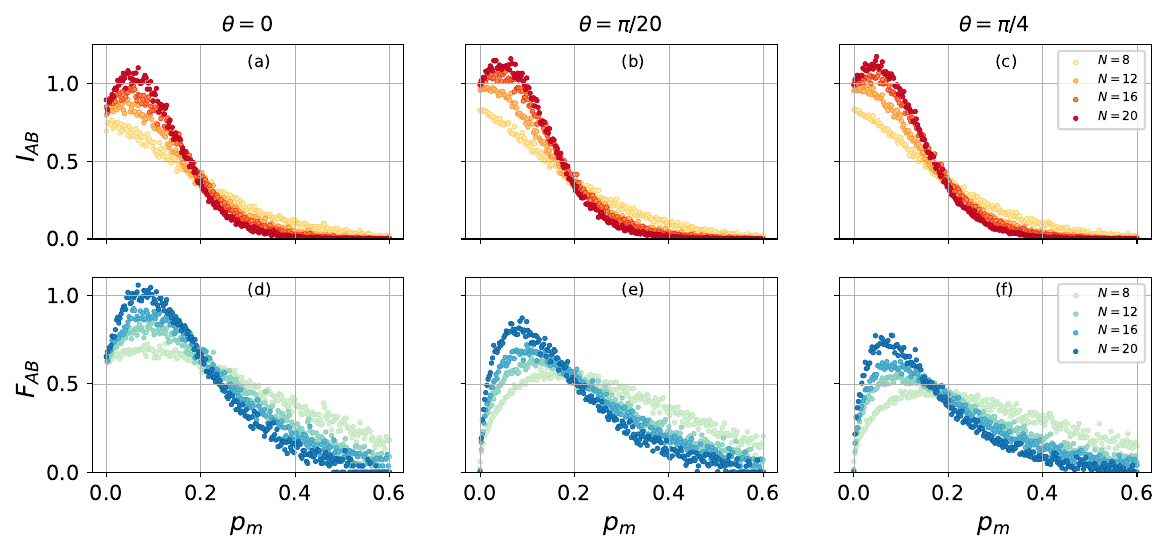}
    \caption{The upper panels correspond to averaged mutual information for different numbers of qubits as a function of measurement rate, $p_m$ for different levels of non-stabilizerness in the circuit, and the lower panels correspond to the fluctuations of the mutual information for the same system. The number of instances for averaging and fluctuations is 800 shots. 
    (a) The mutual information of measurement-induced circuit where the angle of rotation, $\theta$, in the rotation gates, is zero and the circuit is Clifford as a function of measurement rate, $p_m$. (b) The mutual information of the measurement-induced circuit where the angle of rotation, $\theta$, is $\pi/20$ and magic is at an intermediate level, as a function of measurement rate, $p_m$. (c) The mutual information of the measurement-induced circuit where the angle of rotation, $\theta$, is $\pi/4$ and magic is maximum per rotation gate, as a function of measurement rate, $p_m$. (d) The fluctuations of the mutual information for Clifford gates, as a function of measurement rate, $p_m$. (e) The fluctuations of mutual information where the angle of rotation, $\theta$ is $\pi/20$ as a function of measurement rate, $p_m$. (f) The fluctuations of mutual information where the angle of rotation, $\theta$ is $\pi/4$ as a function of measurement rate, $p_m$.  }
    \label{fig:fig4}
\end{figure*}

\section{Effect of magic on measurement-induced phase transition}

In recent years, there have been diverse and rich studies on the entanglement phase transition in random quantum circuits \cite{Li_2018,Li_2019,skinner_2019,Chan_2019,Jian_2020,Gullans_2020,Fisher_2023,Cao_2019,Bao_2020,Zabalo_2020,Tang_2020,Fuji_2020,Iaconis_2020,Lang_2020,Nahum_2021,Lavasani_Alavirad_Barkeshli_2021,Sang_2021,Lunt_2020,Szyniszewski_2020,Gullans_2021,Alberton_2021,Turkeshi_2020,Ippoliti_2021,Yaodong_2021,Fan_2021,Lavasani_2021,Regemortel_2021,Claeys_2022,Agrawal_2022,Block_2022,Noel_Niroula_Zhu_Risinger_Egan_Biswas_Cetina_Gorshkov_Gullans_Huse_etal._2022,Koh_Sun_Motta_Minnich_2023,Oshima_2023}. 
Interestingly, the detection of these phase transitions typically relies on the study of the behaviour of the mutual information in the random quantum circuit \cite{Oshima_2023,Ali_Ghorboon_2023}.
These studies mainly focused on two types of circuits. One type is the Clifford circuits, because of their scalability which is guaranteed by Gottesman-Knill theorem \cite{gottesman1998heisenberg,Aaronson_gottestman_2004} so they are widely used to study the phase transition more accurately \cite{Li_2018,Li_2019,Gullans_2020,Lavasani_2021,Sang_2021,Gullans_2021,Turkeshi_2020,Ippoliti_2021}. The second case is the Haar random circuits that are being used because of their universality in randomness. The possible drawback of using the Haar structures is that they are difficult to scale up and usually have been studied for small-scale systems of qubits \cite{skinner_2019,Ali_Ghorboon_2023,Chan_2019,Lunt_2020,Szyniszewski_2020}. In this section, we study specifically the effect of magic in entanglement phase transition. Specifically, the two research directions mentioned above either avoid non-stabilizer resources altogether or use Haar randomness. Here, we analyze magic injected into the entanglement phase transition circuits in a controlled way and study its effect through the lens of fluctuations analysis developed in Sec.~\ref{sec:2B}.

It is important to mention that, in parallel, there are ongoing studies in phase transitions in magic as well~\cite{bejan2023dynamical,niroula2024phase,fux2023entanglementmagic,gu2024magicinduced,Leone2024purity,tarabunga2024magictransitionmeasurementonlycircuits}. However the purpose of this section is to show the effect of magic solely on the entanglement phase transition identified by the mutual information.

The circuit structure for measurement-induced phase transition we use here inspired by Ref.~\cite{Li_2018,Li_2019,Gullans_2020,Lavasani_2021,Sang_2021,Gullans_2021,Turkeshi_2020} consists of two-qubit blocks of random Clifford gates on the neighbouring qubits. The block of random Clifford gates is generated using the canonical form for uniformly randomly distributed 2-Clifford gate described in \cite{Bravyi_2021}. Here, in order to inject magic in a controlled way, we follow the two-qubit random Clifford gates by a randomly chosen rotation gate from the set of $\{R_x(\theta),R_y(\theta),R_z(\theta)\}$. Afterwards, we induce the projective measurement on each qubit with probability $p_m$ in the z-basis. This process is followed by another set of random Cliffords this time on the odd pairs of qubits. Here, we introduce a periodic boundary condition by connecting the first and the last qubit. After performing Cliffords on odd pairs of qubit, we again follow by projective measurement (again with probability $p_m$). This process forms one cycle of the circuit, as the schematic structure of the circuit is shown in Fig. \ref{fig:fig3} (a) and is repeated $N_{\textit{cycles}}$ times.

We use the mutual information as defined in Eq.~\ref{eq:1} to identify the entanglement phase transition from area law entanglement to volume law entanglement \cite{Ali_Ghorboon_2023}. The partitions that we consider for mutual information are the first and third quarters of the whole system. First, we confirm that there is an entanglement phase transition present at all by checking the Clifford case with zero injected magic, i.e. $\theta=0$, for system sizes $N=\{8,12,16,20 \}$. We average the measured mutual information over $800$ instances and use the circuit depth of $N_{\textit{cycles}}=125$. Due to the small 
system size, we need to employ finite-size scaling to estimate the critical exponent, $\nu$, and the critical measurement rate, $p_{m,c}$. For this range of system sizes after finite-size scaling, we find that the critical measurement rate, $p_{m,c} = 0.19$ and the critical exponent, $\nu=1.25$ (see also Fig.~\ref{fig:fig3} (b) and (c)).

The next step is to inject magic into the system. By increasing the rotation angle to $\theta=\pi/4$, we inject magic into the system. We set the number of cycles and non-Clifford gates to $N_{cycle}=125$ to ensure we gradually saturate the non-stabilizerness in the random circuit. In order to study non-stabilizerness behaviour of an intermediate case, where magic is non-zero but also not maximal, we repeat the same numerical experiment, but with rotation angle $\theta=\pi/20$. 

We show the result of these simulations in the upper panels of Fig. \ref{fig:fig4}, where we plot mutual information, $I_{AB}$, as a function of measurement probability. $p_m$ for rotation angles $\theta\in\{0,\pi/20, \pi/4\}$. First, we notice that there is no difference in the critical measurement rate, $p_{m,c}$ nor in the critical exponent, $\nu$. However, when we analyze fluctuations of $I_{AB}$, $F_{AB}=(\mathbb{E}[(I^{(2)})^2]-\mathbb{E}[I^{(2)}]^2)^{1/2},$, we immediately observe a reduction in fluctuations with increasing non-stabilizerness. The consequence of this fluctuation reduction is better separation of data from different system sizes, which in turn allows for the phase transition point to be more easily identified from mutual information data.

\section{Discussion and Conclusions}

We have shown, both analytically and numerically, that for t-doped Clifford circuits, the fluctuations of the mutual information are proportional to the non-stabilizerness of the system. This observation creates a direct relation between fluctuations of entropic quantity and magic. In Appendix \ref{sec:analytics} we show that this behaviour is fundamentally different than that of Renyi-$4$ entropy (meaning that our observations are not merely a consequence of the fluctuation being a ``fourth-moment quantity''). The relation of magic to the fluctuation of quantum information quantity (OTOC) was previously observed in~\cite{ahmadi2024quantifying}, which is also related to quantum information scrambling. Additionally, similar fluctuations behaviour has already been observed for entanglement entropy \cite{True_2022,Piemontese_2023,OLIVIERO2021127721}, which could be related to a special case of mutual information. Specifically, when we have a pure state and the bi-partition spans the whole system, mutual information reduces to a scaled entanglement entropy.

We also observed that the injection of magic into the measurement-induced phase transition in random circuits decreases the fluctuations of mutual information around the entanglement phase transition, potentially simplifying the identification of this transition from data. It then of course depends on the experimental platform, whether additional rotation gates are feasible to implement. Since we observed that adding \emph{any} amount of magic is beneficial, presumably the phase of these gates would not need to be implemented with high precision, as long as the gate is outside of the Clifford group.

The main open question going forward is that of large-scale simulation of the random circuits with injection of the entanglement. These types of circuits present a particular challenge for approximate methods: they require relatively long time evolution as well as rapid entanglement growth. One candidate that could possibly go beyond these limitations is the Neural Network Quantum States (NQS) \cite{Wu_2020,Wang_2020} which could be an interesting future research direction.

\section{Author contributions}
AA conceived the project with input from EG. AA wrote the code for mutual information fluctuation analysis, analyzed the data and created the figures. CK wrote the code and created figures for the entanglement phase transition analysis with the help of AA and EG. JH derived the analytical proof of the information fluctuation theorem. AA, JH, and EG wrote the manuscript. EG supervised the project.

\section{Data availability}
A GitLab repository containing this project is available at~\cite{fluctutation_code}. All the data and code to analyze them is available at~\cite{zenodo_code_data}. For the simulations of this paper, we used Qiskit and Cirq~\cite{gadi_aleksandrowicz_2019_2562111,cirq_developers_2021_5182845} simulators.

\section{Acknowledgements}
We acknowledge useful discussions with Thomas E. Spriggs, Mohammed Boky, Ana Silva and Ali G. Moghaddam.
This work is part of the project Engineered Topological Quantum Networks (Project No.VI.Veni.212.278) of the research program NWO Talent Programme Veni Science domain 2021 which is financed by the Dutch Research Council (NWO). JH acknowledges funding from the Dutch Research Council (NWO) through Veni No.VI.Veni.222.331 and the Quantum Software Consortium (NWO Zwaartekracht Grant No.024.003.037).

 % \twocolumngrid
\bibliographystyle{apsrev.bst}
\bibliography{bibliography}
\onecolumngrid

 \newpage
\appendix
\section{Fourth moment of spin measurements}\label{sec:kurtosis}

Inspired by the relation between measuring spin fluctuations of a subsystem and entanglement entropy \cite{Ali_Ghorboon_2022} and mutual information \cite{Ali_Ghorboon_2023,pöyhönen2024scalableapproachmonitoredquantum}, we observed that the fourth moment (Kurtosis) of spin measurement of the subsystem also relates to non-stabilizerness on t-doped circuits.

We considered the circuit of the same structure as in Fig. \ref{fig:fig1}. Instead of measuring entanglement entropy, $S^{(2)}$, of subsystems $A$ and $B$, we measured the total spin of those subsystems, $S_z = \sum_{n \in N_l }s_{n,z}$ where $l \in \{A,B\}$. We define $\textit{Kurt}(S_z)_{A,B}$ as

\begin{equation}
    \textit{Kurt}(S_z)_{A,B} :=\textit{Kurt}(\langle S_z\rangle_{A})+\textit{Kurt}(\langle S_z\rangle_{B}) 
    -\textit{Kurt}(\langle S_z\rangle_{AB}).
\end{equation}

The Kurtosis of a random variable $X$ is the standardized fourth moment, defined as, $\frac{\mathbb{E}[(X-\mu)^4]}{\sigma^4}$ where $\mu$ is the mean and $\sigma$ is the standard deviation.

The results of the simulation in Fig. \ref{fig:figAppendixA} show that there is a linear trend in $-\ln \textit{Kurt}(S_z)_{A,B}$ as a function of the number of T-gates in the circuit. In this case, the number of T gates, $N_T$, and the number of samples are both 500. This linear trend depending on the number of qubits, can be observed for an effective number of T-gates. As it is shown in Fig. \ref{fig:figAppendixA}, for an 8 qubit circuit, this effect can be observed from $N_T=0$ while for a 12 qubit circuit, the effective number of T-gates is $N_T = 5$ and for 16 qubit circuit, it is $N_T = 12$.

One possible explanation for the behaviour of small $N_T$ in Fig. \ref{fig:figAppendixA} for not following the linear trend, could be that for measuring spins in each subsystem for Clifford circuits, we have
\begin{equation}
    \langle S_{z,j} \rangle = \Tr(\sigma_{z,j} \rho)=\sum_ic_i\Tr(\sigma_{z,j} P_i),
    \label{eq:appendixC1}
\end{equation}
given the fact that $\Tr (P_iP_j)=\delta_{ij}$ \cite{Gambetta_2012}, so that we obtain
\begin{equation}
    \langle S_{z,j} \rangle= c_j.
    \label{eq:appendixC2}
\end{equation}

Suppose the Clifford circuit is in the scrambling regime $c_j \approx O(4^{-N})$ for $N\gg1$, $c_j\ll1$. Given all instances of spin measurements being a small value, their kurtosis also becomes small and $-\ln \textit{Kurt} (S_z) >1$. So we expect the linear behaviour of $-\ln \textit{Kurt}(S_z)_{A,B}$ as a function of $N_T$ to fail for the small $N_T$ with respect to the system size. 

\begin{figure}
    \centering
    \includegraphics[width=0.5\linewidth]{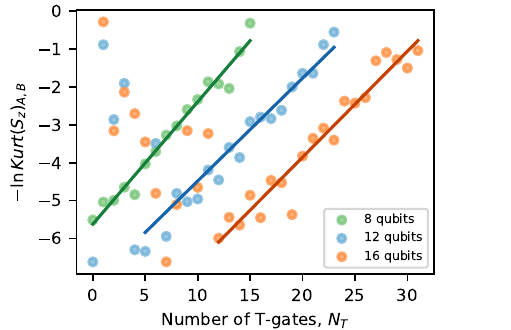}
    \caption{The log of kurtosis of the spin measurements, $-\ln \textit{Kurt}(S_z)_{A,B}$ as a function of the number of T-gates, $N_T$, on the circuit for the number of qubits, $N=8,12,16$ and the number of samples is 500.}
    \label{fig:figAppendixA}
\end{figure}

\section{Analytical treatment of entropy fluctuations}\label{sec:analytics}
In this Supplementary we state a refined version of Eq. (2) in the main text, and prove it. We will make use of the "super-ket" notation: denoting density matrices $\rho$ as $\kett{\rho}$ and observables $E$ as $\braa{E}$, with the trace inner product $\braakett{E}{\rho} = \Tr(E\ct\rho)$. Quantum channels $\Lambda$ act linearly: $\Lambda \kett{\rho} = \kett{\Lambda(\rho)}$. For unitary channels we emphasise the difference between channel and unitary with caligraphic letters, i.e. $\kett{U\rho U\ct} = \mc{U}\kett{\rho}$.  Before we move on to the proof we briefly review relevant representation theory for the unitary and Clifford groups. 

\subsection{Moments of the unitary and Clifford groups}\label{subsec:unitary}
We review the theory of polynomial invariants of the unitary group and the Clifford group. We will focus on quoting results required for the proof. For more detailed explanations see \cite{collins2021weingarten}(unitary group) and \cite{zhu2016clifford,helsen2018representations,gross2021schur,montealegre2021rank}(Clifford group). We begin by discussing the unitary group, and then make adaptations where needed for the Clifford group (focusing in particular on the case of quartic polynomials).
The polynomial invariants of degree $t$ of the unitary group are captured collectively by the \emph{$t$-th moment (super)operator}, which is defined by
\begin{equation}\label{eq:moment op unitary}
  \mathcal M^{U(2^N),(t)} = \int_{U(2^n)} dU \,\mc{U}^{\ot t}.
\end{equation}
This operator (via Schur-Weyl duality) can be expressed in terms of permutation operators: for $\pi \in S_{t}$ we define
\begin{align*}
  R_\pi = r_\pi^{\ot N}
\quad\text{where}\quad
  r_{\pi} = \sum_{x \in \{0,1\}^t} \ket{x_{\pi(1)},\ldots, x_{\pi(t)}} \bra{x_1, \ldots x_t}.
\end{align*}
The moment operator is a projector onto the space spanned by these permutation operators, and we can write
\begin{align}\label{eq:moment via weingarten unitary}
  \mathcal{M}_t = \sum_{\pi',\pi \in S_t} W^{U(2^N), (t)}_{\pi',\pi} \kett{R_{\pi'}}\braa{R_\pi},
\end{align}
where $W^{U(2^N), (t)}_{\pi',\pi} $ is the so-called \emph{Weingarten} matrix, the (pseudo) inverse of the Gram matrix $ G^{U(2^N),(t)}_{\pi,\pi'}:=\braakett{R_\pi}{R_{\pi'}} $ of the permutation operators. A key fact about the Weingarten matrix is that it is diagonally dominant for fixed $t$ and large $n$, we have:
\begin{align}\label{eq:weingarten bound unitary}
  W^{U(2^N),(t)} = 2^{-tN} \parens*{ I + 2^{-N} F },
\end{align}
where~$F$ is a matrix with bounded (as a function of $n$) entries. In particular for $t=4$ we have $\norm{F}\leq 16$.\\

Next we consider the Clifford group. The moment operator has an analogous but more complicated expression, which we will only discuss in detail for $t=4$. We have:
\begin{equation}\label{eq:moment op cliff}
  \mathcal M^{\mathbb C_N,(4)} = \frac{1}{\abs{\mathbb{C}_N}} \sum_{C \in\mathbb{C}_N} \mc{C}^{\ot 4}.
\end{equation}
This operator differs from that of the Haar measure, as the Clifford group is \emph{not} a $4$-design.
For~$n\geq 3$, the dimension of its image is $30$, and a basis is given by the~$24$ permutation operators~$R_\pi$, parameterized by~$\pi \in S_4$,  and six more operators~$R_T$.
The latter can be written in the form~$R_{\hat{\pi}} \Pi_4$, where $\hat{\pi}$ ranges over the subgroup~$S_3 \subseteq S_4$ of permutations acting on the final three subspaces. and
\begin{align}\label{eq:Pi_4}
  \Pi_4 := 2^{-n}(I\tn{4} + X\tn{4} + Y\tn{4} + Z\tn{4})\tn{n}=: \pi_4\tn{n}.
\end{align}
 We will denote this set of six operators as $\hat{S}_3 = \{\pi_4,\pi_4\cdot (23),\pi_4\cdot (34),\pi_4\cdot (24), \pi_4\cdot (234),\pi_4\cdot (324)\}$ where we use the dot to emphasise multiplication. Following convention we will denote the total set of $30$ operators as $\Sigma_{4,4} = S_4\cup \hat{S}_3$
In terms of these operators the moment operator~\eqref{eq:moment op cliff} is given by
\begin{align}\label{eq:moment via weingarten clifford}
  \mathcal{M}^{\mathbb C_N,(4)} = \sum_{T',T \in \Sigma_{4,4}} W^{\mathbb C_N,(4)}_{T',T} \kett{R_{T'}}\braa{R_T}.
\end{align}
where $W^{\mathbb C_N,(4)}$ is the \emph{Clifford-Weingarten matrix}.
We note that $W^{\mathbb C_N,(4)}$ is also diagonally dominant for large~$N$. In particular for $t=4$ we have
\begin{align}\label{eq:weingarten bound}
  W^{\mathbb C_N,(4)}= 2^{-4N} \parens*{ I + 2^{-N} F },
\end{align}
where $\norm{F} \leq 16$.

Finally we will need a lemma from \cite{helsen2023thrifty} that characterises the action of the $T$-gate on the commutant of the Clifford group:
\begin{lemma}\label{lem:t_gate}
Let $\mc{T}$ denote the quantum channel acting by the $\T$-gate 
$\T = \begin{psmallmatrix} 1& 0 \\ 0 & e^{i\pi/4}\end{psmallmatrix}$ 
on the first qubit of an $N$-qubit state.
Then we have, for every $\pi,\pi'\in\hat{S}_3$, that
\begin{align*}
  \braa{R_\pi \Pi_4} \mc{T}^{\ot4} \kett{R_{\pi'} \Pi_4} \begin{cases}
  =(2^4 - 4) 2^{4(N-1)} = \frac34 2^{4N} & \text{ if  }\; \pi=\pi', \\
  \leq (2^3-4) 2^{3(N-1)} = \frac12 2^{3N} & \text{ if  }\; \pi\neq \pi'.
\end{cases}
\end{align*}
\end{lemma}
\subsection{Main theorem}
With the necessary representation theory reviewed we can prove the relation in Eq. $2$ in the main text. For precision's sake we restate the result as a theorem. 
\begin{theorem}
Let $\psi_{N_T}$ be the $N$-qubit state generated by the application of $N_T$ $T$ gates (on a fixed qubit) interspersed with random $N$-qubit Clifford gates, and let $\rho_{AB}$ be its reduced state on $N/2$ qubits. Averaged over the random Clifford gates we have
\begin{align}
\tilde{S}_{AB}^{(2)} &:= -\log\big(\mathbb{E}\Tr(\rho_{AB}^2)\big)  =  1 + N/2 + O(2^{-N})\\
\tilde{\delta} &:=\big(-\log\big(\mathbb{E}\Tr(\rho_{AB}^2)^2\big) -(\tilde{S}_{AB}^{(2)})^2\big) = \bigg(\frac{3}{4}\bigg)^{N_T} + O(2^{-N}).
\end{align}
\end{theorem}
\begin{proof}
We begin by calculating $\tilde{S}_{AB}^{(2)} $. We have
\begin{align}
\mathbb{E}\Tr(\rho_{AB}^2)\big) = \frac{1}{|\mathbb{C}_N|} \sum_{C_0,\ldots, C_{N_T}}
\braa{\mathbb{F}_{AB}\otimes \mathbb{I}_{\overline{AB}} } 
\mc{C}_{N_T}\tn{2}\mc{T}\tn{2} \cdots \mc{T}\tn{2} \mc{C}_0\tn{2}\kett{0\tn{2}},
\end{align}
where $\mathbb{F}$ is the permutation operator exchanging two copies of the $AB$ subsystem. Here we used the trace identities $\tr(A)^2 = \tr(A\tn{2})$ and $\tr(A^2) = \tr(\mathbb{F}A\tn{2})$. We can simplify this equation by noting that the Clifford group is a $2$ design, and hence
\begin{equation}
 \frac{1}{\abs{\mathbb{C}_N}} \sum_{C \in\mathbb{C}_N} \mc{C}^{\ot 2} = \mathcal M^{U(2^N),(2)},
 \end{equation}
 where $\mathcal M^{U(2^N),(2)}$ is the quadratic moment operator of the \emph{unitary group}. By Haar invariance, the action of $\mc{T}\tn{2}$ is easily seen to be absorbed, and since $\mathcal M^{U(2^N),(2)}$ is a projector, we have 
 \begin{equation}
 \mathbb{E}\Tr(\rho_{AB}^2) =\braa{\mathbb{F_{AB}}}\mathcal M^{U(2^N),(2)} \kett{0\tn{2}}.
 \end{equation}
This already shows that $\tilde{S}_{AB}^{(2)} $ is independent of $N_T$. We can calculate the associated value exactly by using the well known formula (see e.g. \cite{gross2021schur})
\begin{equation}
\mathcal M^{U(2^N),(2)} \kett{0\tn{2}} = \frac{1}{2^N(2^N+1)} \sum_{\pi\in S_2}\kett{R_\pi}.
\end{equation}
Using the fact that $\mathbb{F}_{AB} = r_{(12)}\tn{N_{AB}}$ and $R_\pi = r_{\pi}\tn{N}$ separate across qubits we can now calculate
\begin{equation}
 \mathbb{E}\Tr(\rho_{AB}^2) = \frac{1}{2^N(2^N+1)} \sum_{\pi\in S_2}\tr(\mathbb{F} r_{\pi})^{N_{AB}} \tr(r_{\pi})^{N - N_{AB}}.
 \end{equation}
By direct calculation we have $\tr(r_{(12)}) = 2$ and $\tr(r_{e}) = 4$. With a little calculus we thus get
\begin{equation}
 \mathbb{E}\Tr(\rho_{AB}^2) =\frac{2^{N_{AB}} 4^{N-N_{AB}} + 2^{N- N_{AB}}4^{N_{AB}} }{2^N(2^N+1)}.
 \end{equation}
 Using $N_{AB} = N/2$ we can can see that 
 \begin{equation}
 \mathbb{E}\Tr(\rho_{AB}^2) = 2 \, 2^{-N/2} + O(2^{-N}).
 \end{equation}
Calculating the variance term $\mathbb{E}\Tr(\rho_{AB}^2)\big)^2 $ is similar but messier. Again through trace identities we obtain
\begin{equation}
\mathbb{E}\Tr(\rho_{AB}^2)\big)^2  = \frac{1}{|\mathbb{C}_N|^{N_T+1}} \sum_{C_0,\ldots, C_{N_T}}
\braa{(\mathbb{F}_{AB}\otimes \mathbb{I}_{\overline{AB}})\tn{2} } 
\mc{C}_{N_T}\tn{4}\mc{T}\tn{4} \cdots \mc{T}\tn{4} \mc{C}_0\tn{4}\kett{0\tn{4}}.
\end{equation}
Since the Clifford group is not a $4$-design we can no longer simplify this expression. Instead we directly insert
\begin{equation}
 \frac{1}{\abs{\mathbb{C}_N}} \sum_{C \in\mathbb{C}_N} \mc{C}^{\ot 4} = \mathcal M^{\mathbb{C}_N,(4)} = \sum_{T',T \in \Sigma_{4,4}} W^{\mathbb C_n,(4)}_{T',T} \kett{R_{T'}}\braa{R_T}.
 \end{equation}
 Defining the matrix
 \begin{equation}
 Q_{T,T'} = \braa{R_T}\mc{T} \kett{R_{T'}},
 \end{equation}
 and the vectors 
 \begin{align}
 v_T &= \braakett{(\mathbb{F}_{AB}\otimes \mathbb{I}_{\overline{AB}})\tn{2} }{R_T}\\
 u_T &= \braakett{R_T}{0\tn{4}},
\end{align}
we can write the above expression as a matrix vector inner product
\begin{equation}
\mathbb{E}\Tr(\rho_{AB}^2)\big)^2  = v\ct WQW\ldots QW u = v\ct (WQ)^{N_T} Wu.
\end{equation}
At this point it is worth noting that $u$ is the all ones vector, and that the vector $v$ has the property
\begin{equation}
v_{T} \;\;\;\;\;\;\begin{cases} &=  2^{3N}\hspace{2em} \text{    if    } T \in \{ e, (12)(34), (12),(34), \pi_4\cdot(34)\},\\
&\leq 2^{5N/2}\hspace{1.2em} \text{    otherwise},
\end{cases}
\end{equation}
where we used explicitly that $N_{AB} = N/2$ and that $\mathbb{F}_{AB}\tn{2} = r_{(12)(34)}\tn{N_{AB}}$.
Next we make some approximations valid in the large $N$ regime. We will use big O notation to suppress constants that can in principle be calculated. We have that
\begin{equation}
W = 2^{-4n}I + O(2^{-5N})
\end{equation}
and (via lemma $1$) that
\begin{equation}
Q = 2^{4N}\begin{pmatrix} I_{24} & 0 \\ 0& \frac{3}{4}I_6\end{pmatrix} + O(2^{3N}).
\end{equation}
Consequently we have 
\begin{align}
v\ct (WQ)^{N_T}Wu &= 2^{-4N} v\ct \begin{pmatrix} I_{24} & 0 \\ 0 &\frac{3}{4}I_6\end{pmatrix} u +  O(2^{-2N})\\
&=2^{-4N}\bigg(4\, 2^{3N} + \bigg(\frac{3}{4}\bigg)^{N_T}2^{3N}\bigg)+  O(2^{-2N})\\
&=\bigg(4+ \bigg(\frac{3}{4}\bigg)^{N_T}\bigg) 2^{-N} +  O(2^{-2N}),
\end{align}
which finishes this part of the calculation.
With this in hand we can compute $\tilde{\delta}$:
\begin{equation}
\tilde{\delta} = \big(-\log\big(\mathbb{E}\Tr(\rho_A^2)^2\big) -(\tilde{S}_A^{(2)})^2\big) \approx  \log\bigg(1+\bigg(\frac{3}{4}\bigg)^{N_{T}}  O(2^{-N})\bigg)  \approx  \bigg(\frac{3}{4}\bigg)^{N_{T}} + O(2^{-N}),
\end{equation}
where the last approximation is the first order of the Taylor expansion of $\log(1+x)$.
\end{proof}

\begin{figure}[h!]
    \centering    \includegraphics[width=0.5\linewidth]{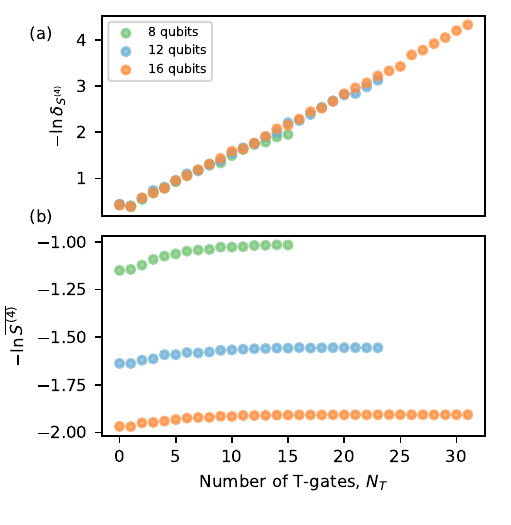}
    \caption{(a)  The log of fluctuations of Renyi-$4$ entropy, $-\ln \delta_{S^{(4)}}$ as a function of the number of T-gates, $N_T$,on the circuit. (b)The log of the average of Renyi-$4$ entropy, $-\log \overline{S^{(4)}}$ as a function of the number of T-gates, $N_T$ on the circuit. The number of qubits, $N=8,12,16$ and the number of samples is 500.}
    \label{fig:figS4}
\end{figure}
\begin{figure}
    \centering
    \includegraphics[width=0.5\linewidth]{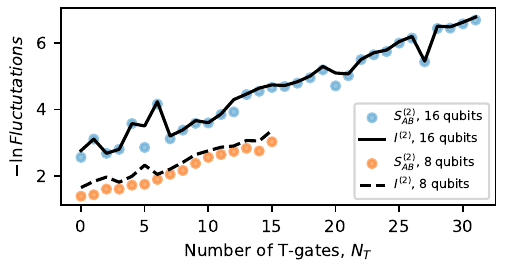}
    \caption{  The log of fluctuations of mutual information $-\ln \delta_{I^{(2)}}$ and Renyi-$4$ entropy, $-\ln \delta_{S^{(2)}}$ as a function of the number of T-gates, $N_T$,on the circuit for the subsystem size $N_A=N_B=N/8$. The number of qubits, $N=8,16$ and the number of samples is 500.}
    \label{fig:figN8}
\end{figure}

The key detail of this calculation is the extra term appearing at order $2^{3N}$ in the expression of the vector $v$. This term is due to the generator $\pi_4 \cdot(34)$ which is present in the commutant of the Clifford group and not that of the unitary group. We also see that due to lemma $1$ it vanishes quickly with increasing $N_T$, as we also observe numerically. It is important to note that this is not merely a consequence of the fluctuation of the Renyi-$2$ entropy being a quartic invariant. In other quartic invariants such as the Renyi-$4$ entropy, a similar calculation can be made but there the contributions from the non-permutation generators are all suppressed by a factor of $2^{-N}$ making them invisible even at moderate qubit numbers. We confirm this observation numerically and show the results in Fig. \ref{fig:figS4} (b): The averaged instances of Renyi-$4$ entropy do not show the linear behaviour, however, the fluctuations of Renyi-$4$ entropy do in Fig. \ref{fig:figS4} (a) showing the behavior is encoded in the fluctuations, not in the order of Renyi entropy. Similarly, it is vital that $|AB| = N/2$. If $|AB|$ is substantially smaller or larger than this (e.g. $|AB|=N/8$), then the contribution of the generator $\pi_4 \cdot(34)$ becomes subleading even in the fluctuation of the Renyi-$2$ entropy. A numerical illustration of how different subsystem sizes influence the result is shown in Fig. \ref{fig:figN8} for the subsystem size of $|AB|=N/4$, where we see that a gap opened for $8$-qubit and $16$-qubit system sizes, thus making the proportionality between mutual information fluctuations and the number of T-gates system size dependent.

\end{document}